\documentclass[12pt]{amsart}
\usepackage{amsmath}
\usepackage{amssymb}
\usepackage[latin1]{inputenc}
\usepackage[T1]{fontenc}
\usepackage{graphicx}
\newtheorem{lemma}{Lemma}

\newtheorem{definition}{Definition}

\newtheorem{proposition}{Proposition}

\newcommand{\tr}{{\rm tr }}
\newcommand{\lh}{\mathcal{L(H)}}

\newcommand{\vp}{\varphi}

\newcommand{\C}{\mathbb{C}}

\newcommand{\N}{\mathbb{N}}

\newcommand{\R}{\mathbb{R}}

\newcommand{\be}{\begin{equation}}
\newcommand{\eeq}{\end{equation}}
\newcommand{\bet}{\begin{equation*}}
\newcommand{\eeqt}{\end{equation*}}
\newcommand{\bea}{\begin{eqnarray}}
\newcommand{\eeqa}{\end{eqnarray}}
\newcommand{\beat}{\begin{eqnarray*}}
\newcommand{\eeqat}{\end{eqnarray*}}

\newcommand{\h}[1]{\mathcal{#1}}
\newcommand{\hil}{\mathcal{H}}

\def\<{\langle}
\def\>{\rangle}

\begin{document}

\title{Inefficient eight-port homodyne detection and covariant phase space observables}

\author{Pekka Lahti}
\address{Turku Centre for Quantum Physics, Department of Physics and Astronomy, University of Turku, FI-20014 Turku, Finland}
\email{pekka.lahti@utu.fi}
\author{Juha-Pekka Pellonp\"a\"a}
\address{Turku Centre for Quantum Physics, Department of Physics and Astronomy, University of Turku, FI-20014 Turku, Finland}
\email{juha-pekka.pellonpaa@utu.fi}
\author{Jussi Schultz}
\address{Turku Centre for Quantum Physics, Department of Physics and Astronomy, University of Turku, FI-20014 Turku, Finland}
\email{jussi.schultz@utu.fi}
\begin{abstract}
We consider the quantum optical eight-port homodyne detection scheme in the case that each of the associated photon detectors is assigned with a different quantum efficiency. We give a mathematically rigorous and strictly quantum mechanical proof of the fact that the measured observable (positive operator measure) in the high-amplitude limit is a smearing of the covariant phase space observable related to the ideal measurement. The result is proved for an arbitary parameter field. Furthermore, we investigate some properties of the measured observable. In particular, we show that the state distinguishing power of the observable is not affected by detector inefficiencies. 

\end{abstract}
\maketitle

\section{Introduction}
The eight-port homodyne detection scheme  has been investigated extensively ever since it was introduced in the realm of quantum optics. The significance of this scheme comes from the fact that it provides a means to study many fundamental questions in quantum mechanics. Among these are the problems of quantum state reconstruction and approximate joint measurements of quadrature observables. The usefulness of this setup is due to the fact that it provides a quantum optical realization of the measurement of any covariant phase space observable \cite{eightport}. With regard to the aforementioned problems, these observables are of great importance. On one hand, since the work of \cite{Prugo}, a large class of covariant phase space observables are known to possess the property that the measurement outcome statistics determine the state uniquely. On the other hand, the quadrature observables are approximately jointly measurable exactly when there exists a covariant phase space observable which is their approximate joint observable \cite{Carmeli}. Thus, it is natural to investigate the detailed structure of the observable measured with this specific scheme.

Since any realistic measurement involves detectors with non-unit quantum efficiencies, it is important to study also the effects of detector inefficiencies in detail. As reported in the recent review on single-photon detectors \cite{Hadfield}, the efficiencies of available detectors range from very high to as low as a few percents. It is therefore clear that in most cases the effect of inefficiencies is far from being negligible. In the eight-port homodyne detection scheme, it was shown in \cite{Leonhardt1993} that with the specific choice of a vacuum parameter field and an overall quantum efficiency for the detectors, the measured probability distribution is a smoothed version of the $Q$-function of the signal field. The smoothing is caused by a Gaussian convolution which is due to the precence of the non-unit quantum efficiencies. Up to our knowledge, this analysis has not yet been done in the case of an arbitrary parameter field, or with different quantum efficiencies for each of the detectors.

The purpose of this paper is to give a mathematically rigorous derivation of the high-amplitude limit observable measured with an inefficient eight-port homodyne detector. The derivation is done strictly within the framework of quantum mechanics without any classicality assumptions. The result is that whenever detector inefficiencies are present, the measured observable is a smearing of the ideal one. Furthermore, we study some basic properties of the measured observable. In particular, we find that the state distinguishing power does not depend on the associated quantum efficiencies. More specifically, we show that the measurement statistics of the ideal observable can always be reconstructed from the smeared statistics. The paper is organized as follows. We start by giving the basic framework for our study in section \ref{preli}. In section \ref{measurement} we derive the high-amplitude limit observable. First, we consider the high-amplitude limit in an inefficient balanced homodyne detector, and then use the results to obtain the measured observable in an inefficient eight-port homodyne detector in the high-amplitude limit. The basic properties of the high-amplitude limit observable are studied in section \ref{properties}, and the conclusions are given in section \ref{conclusions}.

\section{Preliminaries}\label{preli}
Let $\h H$ be a complex separable Hilbert space associated with a single mode electromagnetic field, and let $\{ \vert n\rangle \vert n\in\N  \}$ be an orthonormal basis of $\h H$. Let $a^*$, $a$ and $N$ denote the creation, annihilation and number operators associated with this basis. Let $\lh$ and $\h T(\hil)$ denote the sets of bounded and trace class operators on $\hil$. The states of the system are represented by positive trace class operators with unit trace, density operators, and the pure states correspond to the one-dimensional projections $P[\vp] =\vert \vp\rangle\langle\vp\vert$, $\vp\in\hil$, $\Vert\vp\Vert =1$. Among the pure states are the coherent states $\{ \vert z\rangle\vert z\in\C\}$ defined by
$$
\vert z\rangle = e^{-\frac{\vert z\vert^2}{2}} \sum_{n=0}^\infty \frac{z^n}{\sqrt{n!}}\vert n\rangle.
$$
For $z=\frac{1}{\sqrt{2}}(q+ip)$, the corresponding coherent state $\vert z\rangle$ has the position representation
$$
 \psi_z(x) = \left(\frac{1}{\pi}\right) ^{1/4} e^{-i\frac{qp}{2} } e^{ipx} e^{-\frac{1}{2} (x-q)^2},
$$
 and the subspace $\mathcal{D}_{\textrm{coh}} =\textrm{lin} \{ \vert z\rangle \vert z\in\C \}$ is dense in $\hil$.

The observables are represented by normalized positive operator measures. Among these are the standard quadrature observables $\mathsf{Q}, \mathsf{P}:\h B(\R)\rightarrow\lh$, where $\h B(\R)$ stands for the Borel $\sigma$-algebra of subsets of $\R$. That is, $\mathsf{Q}$ and $\mathsf{P}$ are the spectral measures of the quadrature operators $Q =\frac{1}{\sqrt{2}} (\overline{a^* +a})$ and $P =\frac{i}{\sqrt{2}} (\overline{a^* -a})$, where the bar stands for the closure of an operator. For each $\theta\in [0,2\pi)$ we define the rotated quadrature observable $\mathsf{Q}_\theta:\h B(\R)\rightarrow \lh$ by
$$
\mathsf{Q}_\theta (X) =e^{i\theta N} \mathsf{Q} (X) e^{-i\theta N}, \qquad X\in\h B(\R),
$$
so that in particular $\mathsf{Q}_0 =\mathsf{Q}$ and $\mathsf{Q}_{\pi/2} =\mathsf{P}$. For each positive trace class operator with unit trace $S$, we define the phase space observable $\mathsf{G}^S:\h B(\R^2)\rightarrow \lh$ by
\begin{equation}\label{covariant}
\mathsf{G}^S(Z) =\frac{1}{2\pi} \int_Z W_{qp} S W_{qp}^* \, dqdp, \qquad Z\in\h B(\R^2), 
\end{equation}
where  $W_{qp} =e^{i\frac{qp}{2}} e^{-iqP} e^{ipQ}$ is the Weyl operator. The operator $S$ is called the generating operator of the observable. The mapping $(q,p)\mapsto W_{qp}$ is an irreducible projective unitary representation of $\R^2$, and each $\mathsf{G}^S$ is covariant with respect to $W_{qp}$ in the sense that
$$
W_{qp} \mathsf{G}^S(Z) W_{qp}^* =\mathsf{G}^S (Z+ (q,p))
$$
for all $Z\in\h B(\R^2)$ and $(q,p)\in\R^2$. Furthermore, each covariant phase space observable is of the form \eqref{covariant} for some generating operator $S$ \cite{Holevo, Werner} (for recent alternative proofs, see \cite{Cassinelli, Kiukas}).

For a quantum system in a state $\rho$, the measurement statistics of an observable $\mathsf{E}$ is given by the probability measure $Z\mapsto \mathsf{E}_\rho(Z) = \tr{[\rho\mathsf{E}(Z)]}$. It follows that for each phase space observable $\mathsf{G}^S$ the associated probability measure has the density $g^S_{\rho}(q,p) =\frac{1}{2\pi}\tr[ \rho W_{qp} S W_{qp}^*]$. For a pure state $P[\vp]$, we use the notation $\mathsf{E}_\vp$ for the probability measure related to the observable $\mathsf{E}$, and the notation $\mathsf{G}^{\vp}$ for the observable generated by $P[\vp]$. Any two observables are informationally equivalent if their ability to distinguish between states is equal. If the measurement statistics of an observable determine the state uniquely, the observable is said to be informationally complete.

\section{Measurement scheme}\label{measurement}

\subsection{Inefficient balanced homodyne detector}

\begin{figure}
\includegraphics[width=10cm]{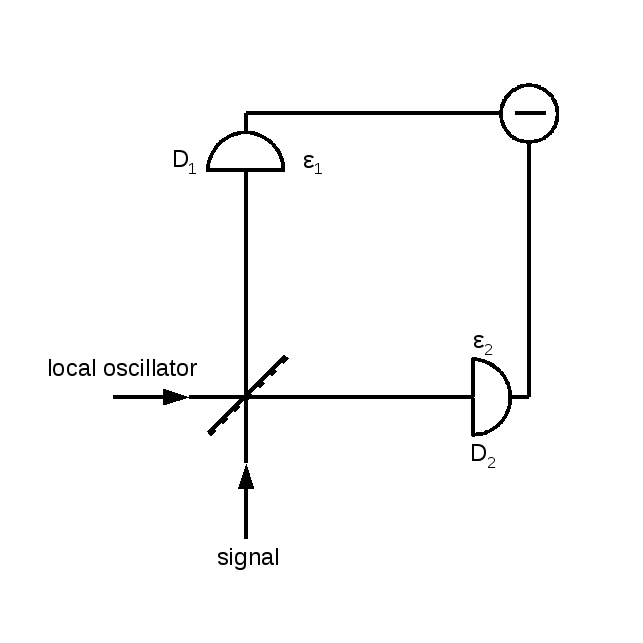}
\caption{Balanced homodyne detector}\label{nonideal}
\end{figure}

The balanced homodyne detector involves two modes, the signal field with the Hilbert space $\hil$ and an auxiliary field of the local oscillator with the Hilbert space $\hil_{\textrm{aux}}$. We denote by $\rho$ the state of the signal field and the auxiliary field is in the coherent state $\vert z\rangle$. These fields are coupled via a lossless $50:50$ beam-splitter which is described by a unitary operator $U$ satisfying 
\begin{equation}\label{beamsplitter}
U\vert\alpha\rangle\otimes\vert\beta\rangle =\vert \tfrac{1}{\sqrt{2}}(\alpha -\beta )\rangle \otimes \vert \tfrac{1}{\sqrt{2}} (\alpha +\beta)\rangle
\end{equation}
for all $\alpha,\beta\in\C$. Here the first term in the tensor product refers to the signal field, and the second term to the auxiliary field. The scheme involves two photon detectors $D_1$ and $D_2$ with quantum efficiencies $\epsilon_1$ and $\epsilon_2$, respectively. With these efficiencies, each of the detectors now measures the smeared photon number, given by the detection observable (see, for instance, \cite[pp. 79-83]{Leonhardt} or \cite[pp. 177-180]{OQP})
\begin{equation}\label{unsharpnumber}
n\mapsto E^{\epsilon_j}_n =\sum_{m=n}^\infty \binom{m}{n} \epsilon_j^n (1-\epsilon_j)^{m-n} \vert m\rangle\langle m\vert.
\end{equation}
We are interested in the scaled photon number differences so that the set of possible measurement outcomes is taken to be
$$
\Omega =\bigg\{ \frac{1}{\sqrt{2}\vert z\vert} \left(\frac{n}{\epsilon_2} -\frac{m}{\epsilon_1}\right) \bigg\vert m,n\in\N \bigg\}.
$$
This specific choice for the scaling is motivated by the fact that it assures that for a coherent signal state the first moment of the probability measure remains finite in the limit $\vert z\vert \rightarrow \infty$. The detection statistics is thus represented by the observable $\mathsf{E}_{\epsilon_1,\epsilon_2} :\h B(\R)\rightarrow \mathcal{L} (\hil \otimes\hil_{\textrm{aux}} )$, 
$$
\mathsf{E}_{\epsilon_1,\epsilon_2} (X) = \sum_{X} E^{\epsilon_1}_m \otimes E^{\epsilon_2}_n
$$
where the summation is now over those $m,n\in\N$ for which $\frac{1}{\sqrt{2}\vert z\vert} \left(\frac{n}{\epsilon_2} -\frac{m}{\epsilon_1}\right) \in X$. The signal observable  $\mathsf{E}^z_{\epsilon_1,\epsilon_2} :\h B(\R)\rightarrow \mathcal{L} (\hil  )$ measured with this setup is now completely determined by the relation
$$
\textrm{tr}[\rho \mathsf{E}^z_{\epsilon_1,\epsilon_2} (X)] =\textrm{tr}[U\rho \otimes \vert z\rangle\langle z\vert U^* \mathsf{E}_{\epsilon_1,\epsilon_2} (X)]
$$
for all states $\rho$ and all $X\in\h B(\R)$, that is, the observable can be written as 
$$
 \mathsf{E}^z_{\epsilon_1,\epsilon_2} (X) = V_z^* U^* \mathsf{E}_{\epsilon_1,\epsilon_2} (X) U V_z,\qquad X\in\h B(\R)
$$
where $V_z :\hil\rightarrow \hil \otimes \hil_{\textrm{aux}}$ is the linear isometry $\varphi\mapsto \varphi\otimes \vert z\rangle$.

To consider rigorously the high-amplitude limit $\vert z\vert\rightarrow\infty$ in this measurement scheme, we need to be specific about what we mean by the limit of the associated observables. First of all, we recall that a sequence $(p_k)_{k\in\N}$ of probability measures $p_k:\h B(\R^n)\rightarrow [0,1]$ converges weakly to a probability measure $p:\h B(\R^n)\rightarrow [0,1]$ if $\lim_{k\rightarrow\infty} \int f(x)\, dp_k(x) =\int f(x)\, dp(x)$  for all bounded continuous functions $f:\R^n\rightarrow\R$. According to \cite[Theorem 2.1]{Billingsley}, the weak convergence is equivalent to the condition $\lim_{k\rightarrow\infty} p_k (X) =p(X)$ for all $X\in \h B(\R^n)$ such that $p(\partial X)=0$, where $\partial X$ denotes the boundary of  $X$. This is the motivation for the following definition, used also in \cite{moment}.
\begin{definition}
A sequence $(\mathsf{E}^k )_{k\in\N}$ of observables $\mathsf{E}^k:\h B(\R^n)\rightarrow \hil$ converges to an observable $\mathsf{E}:\h B (\R^n)\rightarrow\hil$ weakly in the sense of probabilities if
$$
\lim_{k\rightarrow\infty} \mathsf{E}^k (X) =\mathsf{E} (X)
$$
in the weak operator topology for all $X\in\h B(\R^n)$ such that $\mathsf{E} (\partial X)=0$.
\end{definition}
Several equivalent conditions for this convergence are given in \cite[Proposition 10]{moment}. In particular, this convergence happens if and only if there exists a dense subspace $\h D\subset\hil$ such that for all unit vectors $\vp\in\h D$, the corresponding sequence $(\mathsf{E}^k_\vp)_{k\in\N}$ of probability measures converges weakly to $\mathsf{E}_\vp$. Note that since the weak limit of a sequence of probability measures is unique \cite[Theorem 1.3]{Billingsley}, it follows that a sequence of observables can converge to at most one observable weakly in the sense of probabilities. Furthermore, according to the continuity theorem \cite[Theorem 7.6]{Billingsley}, the weak convergence of probability measures is equivalent to the pointwise convergence of the corresponding characteristic functions. We will use these facts with the choice $\h D =\h D_{\textrm{coh}}$ to prove our result.

We fix the phase $\theta \in [0,2\pi)$ of the local oscillator and take an arbitrary sequence $(r_k)_{k\in\N}$ of positive numbers such that $\lim_{k\rightarrow\infty} r_k =\infty$. Let $z_k =r_k e^{i\theta}$, so that we obtain a sequence $(\mathsf{E}^{z_k}_{\epsilon_1,\epsilon_2})_{k\in\N}$ of observables $\h B(\R)\rightarrow\h L(\hil)$. Suppose that $\epsilon_1<1$ or $\epsilon_2 <1$, and define the probability density $f_{\epsilon_1,\epsilon_2}:\R\rightarrow\R$ by 
\begin{equation}\label{tiheys}
f_{\epsilon_1,\epsilon_2} (x) = \sqrt{\tfrac{2\epsilon_1\epsilon_2}{\pi (\epsilon_1 -2\epsilon_1\epsilon_2 +\epsilon_2)}}\, e^{-\frac{2\epsilon_1\epsilon_2}{\epsilon_1 -2\epsilon_1\epsilon_2 +\epsilon_2} x^2}.
\end{equation}
Let $\mu_{\epsilon_1,\epsilon_2}:\h B(\R)\rightarrow[0,1]$ be the probability measure determined by $f_{\epsilon_1,\epsilon_2}$, that is, $\mu_{\epsilon_1,\epsilon_2}(X) =\int_X f_{\epsilon_1,\epsilon_2}(x)\, dx$ for all $X\in \h B(\R)$. We wish to extend the definition of $\mu_{\epsilon_1,\epsilon_2}$ to include also the case of ideal detectors, and thus we define $\mu_{1,1}$ as the Dirac measure concentrated at the origin. We will prove in the next proposition, that the smeared rotated quadrature observable $\mu_{\epsilon_1,\epsilon_2} *\mathsf{Q}_\theta:\h B(\R)\rightarrow\h L(\hil)$ defined as the weak integral
\begin{equation}\label{sumea}
(\mu_{\epsilon_1,\epsilon_2} *\mathsf{Q}_\theta)(X)  =\int \mu_{\epsilon_1,\epsilon_2}(X-x) \, d\mathsf{Q}_\theta(x),\qquad X\in\h B(\R),
\end{equation}
is the high-amplitude limit in this measurement scheme. Note that $\mu_{1,1} \ast \mathsf{Q}_\theta =\mathsf{Q}_\theta$. We start with a lemma.

\begin{lemma}\label{raja}
For all $a,b\in\R\setminus \{ 0 \}$ we have 
$$
\lim_{x\rightarrow\infty}\left[ax^2\left(1-e^{-\frac{i}{ax}}\right) +bx^2\left(1-e^{\frac{i}{bx}}\right)\right]=\frac{1}{2}\left(\frac{1}{a} +\frac{1}{b}\right).
$$
\end{lemma}
\begin{proof}
Using the change of variables $y=\frac{1}{x}$ and l'Hospital's rule twice we have
\begin{eqnarray*}
 &&\lim_{x\rightarrow\infty} \left[ ax^2\left(1-e^{-\frac{i}{ax}}\right) +bx^2\left(1-e^{\frac{i}{bx}}\right)\right]\\
&=& \lim_{y\rightarrow 0+}\frac{a(1-\cos(y/a)) + b(1-\cos(y/b))}{y^2}  + i \lim_{y\rightarrow 0+} \frac{a\sin(y/a) -b\sin(y/b)}{y^2}\\
&=& \lim_{y\rightarrow 0+} \frac{1}{2} \left( \frac{1}{a}\cos (y/a) + \frac{1}{b} \cos (y/b)\right) + i \lim_{y\rightarrow 0+} \frac{1}{2}\left( \frac{1}{b} \sin (y/b) -\frac{1}{a} \sin (y/a)\right)\\
&=&\frac{1}{2}\left(\frac{1}{a} +\frac{1}{b}\right)
\end{eqnarray*}
\end{proof}

\begin{proposition}\label{homodynelimit}
 For all $\epsilon_1,\epsilon_2 \in (0,1]$ the sequence $(\mathsf{E}^{z_k}_{\epsilon_1,\epsilon_2})_{k\in\N}$ converges to $\mu_{\epsilon_1,\epsilon_2} *\mathsf{Q}_\theta$ weakly in the sense of probabilities.
\end{proposition}
\begin{proof} The case $\epsilon_1 =\epsilon_2 =1$ has been proved in \cite{moment}, so we may assume that $\epsilon_1 <1$ or $\epsilon_2 <1$. We need to show that
 \begin{equation}\label{raja2}
\lim_{k\rightarrow\infty}\int e^{itx}\, d\langle\alpha \vert \mathsf{E}^{z_k}_{\epsilon_1,\epsilon_2}(x) \vert\beta\rangle =\int e^{itx}\, d\langle \alpha\vert \left(\mu_{\epsilon_1,\epsilon_2} *\mathsf{Q}_\theta \right) (x) \vert \beta\rangle
\end{equation}
for all $\alpha, \beta \in\C$ and $t \in\R$. For $t=0$ the equation is clearly true, so we assume now that $t\neq 0$.

First note that for all $\alpha,\beta\in\C$ we have 
\begin{eqnarray*}
\langle \alpha\vert V_{z_k}^* U^* E^{\epsilon_1}_m \otimes E^{\epsilon_2}_n U V_{z_k}\vert\beta\rangle &=& \langle \tfrac{1}{\sqrt{2}}(\alpha -z_k) \vert E^{\epsilon_1}_m \vert \tfrac{1}{\sqrt{2}}(\beta -z_k) \rangle \langle \tfrac{1}{\sqrt{2}}(\alpha +z_k) \vert E^{\epsilon_2}_n \vert \tfrac{1}{\sqrt{2}}(\beta +z_k) \rangle \\
&=& \frac{1}{m!n!} \left( \tfrac{\epsilon_1}{2} (\overline{\alpha} -\overline{z}_k)(\beta -z_k)\right)^m \left(\tfrac{\epsilon_2}{2} (\overline{\alpha} +\overline{z}_k)( \beta +z_k)\right)^n \\
&&\times\, e^{-\frac{1}{2}\vert\alpha\vert^2 -\frac{1}{2}\vert\beta\vert^2 -\vert z_k\vert^2} 
e^{\frac{1}{2} (1-\epsilon_1)(\overline{\alpha} -\overline{z}_k)(\beta -z_k) +\frac{1}{2}(1-\epsilon_2) (\overline{\alpha} +\overline{z}_k)( \beta +z_k)}
\end{eqnarray*}
so that 
\begin{eqnarray*}
 &&\int e^{itx}\, d\langle\alpha \vert \mathsf{E}^{z_k}_{\epsilon_1,\epsilon_2}(x) \vert\beta\rangle = \sum_{m,n=0}^\infty e^{\frac{it}{\sqrt{2}\vert z_k\vert }\left(\frac{n}{\epsilon_2} -\frac{m}{\epsilon_1}\right)} \langle \alpha\vert V_{z_k}^* U^* E^{\epsilon_1}_m \otimes E^{\epsilon_2}_n U V_{z_k}\vert\beta\rangle\\
&=& \sum_{m,n=0}^\infty  \frac{1}{m!n!} \left(e^{-\frac{it}{\sqrt{2}\epsilon_1\vert z_k\vert }}\right)^m \left(e^{\frac{it}{\sqrt{2}\epsilon_2\vert z_k\vert }}\right)^n \left( \tfrac{\epsilon_1}{2} (\overline{\alpha} -\overline{z}_k)(\beta -z_k)\right)^m \left(\tfrac{\epsilon_2}{2} (\overline{\alpha} +\overline{z}_k)( \beta +z_k)\right)^n \\
&&\times\, e^{-\frac{1}{2}\vert\alpha\vert^2 -\frac{1}{2}\vert\beta\vert^2 -\vert z_k\vert^2} 
e^{\frac{1}{2} (1-\epsilon_1)(\overline{\alpha} -\overline{z}_k)(\beta -z_k) +\frac{1}{2}(1-\epsilon_2) (\overline{\alpha} +\overline{z}_k)( \beta +z_k)}\\
&=&e^{-\frac{1}{2}\vert\alpha\vert^2 -\frac{1}{2}\vert\beta\vert^2 +\overline{\alpha}\beta} e^{-\frac{\epsilon_1}{2}(1- \textrm{exp}(-\frac{it}{\sqrt{2}\epsilon_1\vert z_k\vert }))(\overline{\alpha} -\overline{z}_k)(\beta -z_k)} e^{-\frac{\epsilon_2}{2}( 1-\textrm{exp}(\frac{it}{\sqrt{2}\epsilon_2\vert z_k\vert })) (\overline{\alpha} +\overline{z}_k)( \beta +z_k)}\\
&=& e^{-\frac{1}{2}\vert\alpha\vert^2 -\frac{1}{2}\vert\beta\vert^2 +\overline{\alpha}\beta} e^{-\frac{\epsilon_1}{2}(1- \cos(\frac{t}{\sqrt{2}\epsilon_1 r_k}) +i\sin(\frac{t}{\sqrt{2}\epsilon_1 r_k }))(\overline{\alpha}\beta -r_k (\overline{\alpha}e^{i\theta} +\beta e^{-i\theta}) )} \\
&&\times\,  e^{-\frac{\epsilon_2}{2}((1- \cos(\frac{t}{\sqrt{2}\epsilon_2 r_k}) -i\sin(\frac{t}{\sqrt{2}\epsilon_2 r_k }))(\overline{\alpha}\beta +r_k (\overline{\alpha}e^{i\theta} +\beta e^{-i\theta}) )} e^{-\frac{\epsilon_1}{2}r_k^2(1- \textrm{exp}(-\frac{it}{\sqrt{2}\epsilon_1 r_k})) - \frac{\epsilon_2}{2}r_k^2( 1-\textrm{exp}(\frac{it}{\sqrt{2}\epsilon_2 r_k}))} 
\end{eqnarray*}
Now we may use lemma \ref{raja} and standard limit results  for trigonometric functions to calculate
$$
\lim_{k\rightarrow\infty}\int e^{itx}\, d\langle\alpha \vert \mathsf{E}^{z_k}_{\epsilon_1,\epsilon_2}(x) \vert\beta\rangle = e^{-\frac{1}{2}\vert\alpha\vert^2 -\frac{1}{2}\vert\beta\vert^2 +\overline{\alpha}\beta}   e^{\frac{it}{\sqrt{2}} (\overline{\alpha}e^{i\theta} +\beta e^{-i\theta})}   e^{-\frac{t^2}{8}\left(\frac{1}{\epsilon_1} +\frac{1}{\epsilon_2}\right)}.
$$
We still need to show that this is the right-hand side of equation \eqref{raja2}.

Since for all $X\in \h B(\R)$ we have
$$
\langle \alpha \vert \mathsf{Q}_\theta (X)\vert\beta\rangle = \langle \alpha \vert e^{i\theta N} \mathsf{Q} (X)e^{-i\theta N} \vert \beta \rangle = \langle e^{-i\theta} \alpha \vert \mathsf{Q} (X) \vert e^{-i\theta}\beta\rangle,
$$
we may express the density of the measure $X\mapsto \langle\alpha\vert \left(\mu_{\epsilon_1,\epsilon_2} \ast \mathsf{Q}_\theta \right) (X)\vert \beta\rangle$ as
$$
x\mapsto \int f_{\epsilon_1,\epsilon_2} (x-y) \, d\langle\alpha' \vert \mathsf{Q}(y) \vert \beta'\rangle,
$$
where $\alpha'=e^{-i\theta}\alpha$ and $\beta' =e^{-i\theta}\beta$. Putting $\alpha' =\frac{1}{\sqrt{2}} (q+ip)$ and  $\beta' =\frac{1}{\sqrt{2}}(u+iv)$ we find that in the position representation 
\begin{eqnarray*}
&&\int e^{itx}\, d\langle \alpha\vert \left(\mu_{\epsilon_1,\epsilon_2} *\mathsf{Q}_\theta \right)(x) \vert \beta\rangle = \int e^{itx}\left( \int f_{\epsilon_1,\epsilon_2} (x-y) \, d\langle\alpha' \vert \mathsf{Q}(y) \vert \beta'\rangle\right)\, dx\\
&=& \frac{1}{\pi}\sqrt{\tfrac{2\epsilon_1\epsilon_2}{\epsilon_1 -2\epsilon_1\epsilon_2 +\epsilon_2}}\int e^{itx}\left(\int e^{-\frac{2\epsilon_1\epsilon_2}{\epsilon_1 -2\epsilon_1\epsilon_2 +\epsilon_2} (x-y)^2} e^{\frac{i}{2}(qp-uv)} e^{iy(v-p)} e^{-\frac{1}{2}(y-q)^2 -\frac{1}{2}(y-u)^2} \, dy\right)\, dx\\
&=& e^{-\frac{1}{2}\vert\alpha\vert^2 -\frac{1}{2}\vert\beta\vert^2 +\overline{\alpha}\beta} e^{\frac{it}{\sqrt{2}}(\overline{\alpha}e^{i\theta} +\beta e^{-i\theta})} e^{-\frac{t^2}{8} \left(\frac{1}{\epsilon_1} +\frac{1}{\epsilon_2}\right)}.
\end{eqnarray*}
It follows that  
$$
\lim_{k\rightarrow\infty}\int e^{itx}\, d\langle\varphi \vert \mathsf{E}^{z_k}_{\epsilon_1,\epsilon_2}(x) \varphi\rangle =\int e^{itx}\, d\langle \varphi\vert \left(\mu_{\epsilon_1,\epsilon_2} *\mathsf{Q}_\theta \right) (x) \varphi\rangle
$$
for all unit vectors $\varphi\in\mathcal{D}_{\textrm{coh}}$, so the claim follows from \cite[Proposition 10]{moment} and the continuity theorem \cite[Theorem 7.6]{Billingsley}.

\end{proof}

\subsection{Inefficient eight-port homodyne detector}

\begin{figure}
\includegraphics[width=12cm]{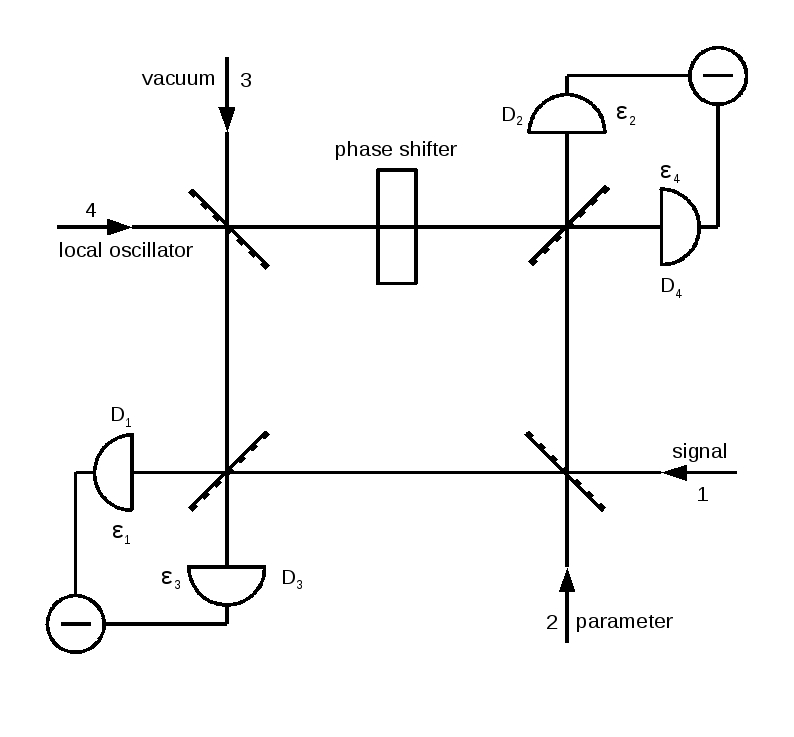}
\caption{Eight-port homodyne detector}\label{detector}
\end{figure}

The eight-port homodyne detector involves four input modes, four $50:50$ beam splitters, a phase shifter, and four photon detectors (see figure \ref{detector}). If $\hil_j$, $j=1,2,3,4$, is the Hilbert space of the $j$th input mode, then the Hilbert space of the entire four mode field is $\hil_1 \otimes \hil_2 \otimes \hil_3 \otimes \hil_4$. We denote by $\rho$  the state of the signal field and by $S$ the state of the parameter field. If the coherent local oscillator is in the state $\vert \sqrt{2}z\rangle$, the initial state of the four-mode field is
$$
\rho\otimes S \otimes \vert 0\rangle\langle 0\vert \otimes \vert \sqrt{2} z\rangle\langle \sqrt{2}z \vert.
$$
In this case we use the notation  $U_{ij}$ for the unitary transform representing the $50:50$ beam splitter. Here the subscripts refer to the primary and secondary input modes, that is, the first and second components of the tensor product in equation \eqref{beamsplitter}. The dashed lines in figure \ref{detector} represent the primary input modes. The phase shifter with phase shift $\phi$ is modelled with the unitary operator $e^{i\phi N}$.

We assign to each detector $D_j$ a quantum efficiency $\epsilon_j\in (0,1]$, so that each detector measures the observable defined in equation \eqref{unsharpnumber}. The detection is represented by the biobservable
$$
(X,Y)\mapsto \mathsf{E}_{\epsilon_1,\epsilon_3}(X)\otimes \mathsf{E}_{\epsilon_2,\epsilon_4} (Y) =\sum_{X,Y} E^{\epsilon_1}_k \otimes E^{\epsilon_2}_l \otimes E^{\epsilon_3}_m \otimes E^{\epsilon_4}_n,
$$
where the summation is now taken over those $k,l,m,n\in\N$ for which $\frac{1}{\sqrt{2}\vert z\vert}\left(\frac{m}{\epsilon_3} -\frac{k}{\epsilon_1}\right) \in X$ and $\frac{1}{\sqrt{2}\vert z\vert}\left(\frac{n}{\epsilon_4} -\frac{l}{\epsilon_2}\right) \in Y$. The state of the entire four-mode field before detection is 
$$
\sigma_{\rho,S,z,\phi} =     U_{13} \otimes U_{24} \left( U_{12}(\rho \otimes S) U_{12}^* \otimes \vert z\rangle\langle z\vert \otimes  \vert ze^{i\phi}\rangle\langle  ze^{i\phi}\vert \right)U_{13}^* \otimes U_{24}^*,
$$
so that the detection statistics are given by the probability bimeasures
$$
(X,Y)\mapsto \textrm{tr} [\sigma _{\rho,S,z,\phi} \mathsf{E}_{\epsilon_1,\epsilon_3}(X)\otimes \mathsf{E}_{\epsilon_2,\epsilon_4} (Y)].
$$
Now there exists a unique signal observable $\mathsf{E}^{S,z,\phi}:\h B(\R^2)\rightarrow \h L(\hil_1)$ such that 
$$
\textrm{tr}[\rho \mathsf{E}^{S,z,\phi} (X\times Y)] =\textrm{tr}[ \sigma _{\rho,S,z,\phi} \mathsf{E}_{\epsilon_1,\epsilon_3}(\tfrac{1}{\sqrt{2}}X)\otimes \mathsf{E}_{\epsilon_2,\epsilon_4} (\tfrac{1}{\sqrt{2}}Y)],
$$
where the scaling has been chosen for later convenience. In order to calculate the high-amplitude limit we wish to express $\mathsf{E}^{S,z,\phi}$ in terms of the unsharp homodyne detection observables $\mathsf{E}^z_{\epsilon_1,\epsilon_3} $ and $\mathsf{E}^{ze^{i\phi}}_{\epsilon_2,\epsilon_4} $. In fact, after simple calculations we find that 
$$
 \textrm{tr}[\rho \mathsf{E}^{S,z,\phi} (X\times Y)] =\textrm{tr}[ U_{12}(\rho \otimes S) U_{12}^* \mathsf{E}^z_{\epsilon_1,\epsilon_3}(\tfrac{1}{\sqrt{2}}X) \otimes \mathsf{E}^{ze^{i\phi}}_{\epsilon_2,\epsilon_4} (\tfrac{1}{\sqrt{2}}Y)]
$$
for all $X,Y\in\h B(\R)$. Denote again $z_k =r_k e^{i\theta}$, where $\theta \in [0,2\pi)$ is fixed and $(r_k)_{k\in\N}$ is an arbitrary sequence of positive numbers such that $\lim_{k\rightarrow\infty} r_k =\infty$. It follows from proposition \ref{homodynelimit} and the boundedness of the associated operators that for all $X,Y\in\h B(\R)$ such that the boundaries $\partial X$ and $\partial Y$ are of zero Lebesgue measure, we have the convergence
$$
\lim_{k\rightarrow\infty} \textrm{tr}[\rho \mathsf{E}^{S,z_k,\phi} (X\times Y)] = \textrm{tr}[ U_{12}(\rho \otimes S) U_{12}^*
\left(\mu_{\epsilon_1,\epsilon_3} \ast \mathsf{Q}_\theta \right) (\tfrac{1}{\sqrt{2}}X) \otimes \left(\mu_{\epsilon_2,\epsilon_4}\ast \mathsf{Q}_{\theta +\phi} \right) (\tfrac{1}{\sqrt{2}}Y)].
$$
Note that the condition of zero Lebesgue measure follows from the fact that each $\mathsf{Q}_\theta$ is unitarily equivalent to $\mathsf{Q}$ which is absolutely continuous with respect to the Lebesgue measure. In particular, we may choose $\theta =0$ and $\phi =\frac{\pi}{2}$ to obtain the limit
$$
\lim_{k\rightarrow\infty} \textrm{tr}[\rho \mathsf{E}^{S,r_k ,\frac{\pi}{2}} (X\times Y)] = \textrm{tr}[ U_{12}(\rho \otimes S) U_{12}^*
\left(\mu_{\epsilon_1,\epsilon_3} \ast \mathsf{Q} \right) (\tfrac{1}{\sqrt{2}}X) \otimes \left(\mu_{\epsilon_2,\epsilon_4}\ast \mathsf{P}\right) (\tfrac{1}{\sqrt{2}}Y)].
$$
Now we still need to find the explicit form of the high-amplitude limit observable.

Let $\mu_{\epsilon}:\h B(\R^2)\rightarrow [0,1]$ be the unique probability measure satisfying 
\begin{equation}\label{mitta}
\mu_{\epsilon } (X\times Y ) = \mu_{\epsilon_1,\epsilon_3} (\tfrac{1}{\sqrt{2}} X) \mu_{\epsilon_2,\epsilon_4 } (\tfrac{1}{\sqrt{2}} Y)
\end{equation}
for all $X,Y\in \h B(\R)$. Here we have chosen a collective symbol $\epsilon$ to represent the  involved quantum efficiencies $( \epsilon_1, \epsilon_2, \epsilon_3, \epsilon_4 )$. This probability measure has a density which we denote by $f_\epsilon$ if and only if both $\mu_{\epsilon_1,\epsilon_3} $ and $\mu_{\epsilon_2,\epsilon_4}$ have densities given by \eqref{tiheys}. In the rest of the paper we will indicate explicitly when we assume the existence of the density $f_\epsilon$. Let $C$ denote the conjugation map $\psi\mapsto (x\mapsto \overline{\psi(x)})$ and let $(r_k)_{k\in\N}$ be as before. The high-amplitude limit observable is now given by the following proposition, in which the smeared phase space observable is defined as a weak integral similar to \eqref{sumea}.

\begin{proposition}\label{smearing}
 The sequence $(\mathsf{E}^{S,r_k ,\frac{\pi}{2}})_{k\in\N}$ converges to $\mu_{\epsilon} *\mathsf{G}^{CSC^{-1}}$ weakly in the sense of probabilities.
\end{proposition}
\begin{proof}
 We begin by showing that 
\begin{equation}\label{kaava1}
 \textrm{tr}[ U_{12}(\rho \otimes S) U_{12}^*
\left(\mu_{\epsilon_1,\epsilon_3} \ast \mathsf{Q} \right) (\tfrac{1}{\sqrt{2}}X) \otimes \left(\mu_{\epsilon_2,\epsilon_4}\ast \mathsf{P}\right) (\tfrac{1}{\sqrt{2}}Y)] = \textrm{tr} [\rho  (\mu_{\epsilon} *\mathsf{G}^{CSC^{-1}}) (X \times Y)]
\end{equation}
for all $X,Y\in \h B(\R^2)$. 

Let $\vp \in \hil_1$ and $\psi \in\hil_2$ be unit vectors. First note that $\mathsf{P} (\cdot ) =F^{-1} \mathsf{Q}(\cdot) F$ where $F$ is the Fourier-Plancherel operator. Furthermore, the relation 
$$
(I\otimes F) U_{12} (\vp \otimes \psi )(x,y)  =\frac{1}{\sqrt{\pi}} \langle W_{\sqrt{2}x,\sqrt{2}y} C\psi\vert \vp\rangle 
$$
holds for all $y\in\R$ and almost all $x\in\R$ (see, e.g., the proof of \cite[Lemma 2]{eightport}). Now a direct calculation shows us that
\begin{eqnarray*}
 && \textrm{tr}[ U_{12}(P[\vp] \otimes P[\psi]) U_{12}^*
\left(\mu_{\epsilon_1,\epsilon_3} \ast \mathsf{Q} \right) (\tfrac{1}{\sqrt{2}}X) \otimes \left(\mu_{\epsilon_2,\epsilon_4}\ast \mathsf{P}\right) (\tfrac{1}{\sqrt{2}}Y)] \\
&=& \langle  (I\otimes F) U_{12} (\vp\otimes\psi ) \vert (\mu_{\epsilon_1,\epsilon_3}  \ast  \mathsf{Q} )(\tfrac{1}{\sqrt{2}} X ) \otimes  (\mu_{\epsilon_2,\epsilon_4}  \ast  \mathsf{Q} )(\tfrac{1}{\sqrt{2}} Y ) (I\otimes F) U_{12} (\vp\otimes\psi) \rangle \\
&=& \int \mu_{\epsilon_1,\epsilon_3} (\tfrac{1}{\sqrt{2}} X -x) \mu_{\epsilon_2,\epsilon_4} (\tfrac{1}{\sqrt{2}}  Y-y) \big\vert ((I\otimes F) U_{12} \vp \otimes \psi )(x,y) \big\vert^2\, dxdy\\
&=& \frac{1}{\pi} \int \mu_{\epsilon_1,\epsilon_3} (\tfrac{1}{\sqrt{2}} X -x) \mu_{\epsilon_2,\epsilon_4} (\tfrac{1}{\sqrt{2}}  Y-y) \big\vert \langle W_{\sqrt{2}x,\sqrt{2}y} C\psi\vert \vp\rangle \big\vert^2 dx dy \\
&=& \frac{1}{2\pi} \int \mu_{\epsilon_1,\epsilon_3} (\tfrac{1}{\sqrt{2}} (X -x')) \mu_{\epsilon_2,\epsilon_4} (\tfrac{1}{\sqrt{2}}  (Y-y'))  \big\vert \langle W_{x',y'} C\psi\vert \vp\rangle \big\vert^2 dx' dy' \\
&=& \frac{1}{2\pi} \int  \mu_\epsilon (X\times Y-(x',y'))  \big\vert \langle W_{x',y'} C\psi\vert \vp\rangle \big\vert^2 dx' dy' \\
&=& \langle \vp\vert (\mu_\epsilon \ast \mathsf{G}^{C\psi}) (X \times Y)\vp\rangle
\end{eqnarray*} 
for all $X,Y\in\h B(\R)$, so that equation \eqref{kaava1} holds for $\rho =P[\vp]$ and $S= P[\psi]$. Since both sides of equation \eqref{kaava1} depend linearly and continuously on $\rho$ and $S$, the validity of the equation in the general case follows by using the spectral representations for $\rho$ and $S$.

Now let $X,Y\in\h B(\R)$ be such that $\partial X$ and $\partial Y$ are of zero Lebesgue measure, so that according to the previous discussion we have the convergence
$$
\lim_{k\rightarrow\infty} \textrm{tr}[\rho\mathsf{E}^{S,r_k ,\frac{\pi}{2}} ( X\times Y)] =\textrm{tr} [\rho (\mu_{\epsilon} *\mathsf{G}^{CSC^{-1}})(X\times Y)]
$$
for any state $\rho$. Since the family of sets of the form $X\times Y$ where the boundaries of $X$ and $Y$ are of zero Lebesgue measure is closed under finite intersections and includes a neighbourhood base of any point $(x,y)\in\R^2$, it follows from \cite[Corollary 1, p. 14]{Billingsley} that for any state $\rho$, the sequence $( \mathsf{E}^{S,r_k ,\frac{\pi}{2}}_\rho )_{k\in\N}$ of probability measures converges weakly to the probability measure $\mu_{\epsilon} *\mathsf{G}^{CSC^{-1}}_\rho$. This completes our proof.
\end{proof}

\section{Some properties of the high-amplitude limit observable}\label{properties}
In \cite{eightport} it was shown that in the case of ideal photon detectors the high-amplitude limit observable is the covariant phase space observable $\mathsf{G}^{CSC^{-1}}$. Proposition \ref{smearing} now implies that the presence of inefficiencies causes a Gaussian smearing of the observable so that the actually measured observable is $\mu_{\epsilon} *\mathsf{G}^{CSC^{-1}}$. In this section we consider some properties of this smeared observable.

The first important observation is given in the next proposition which shows that the covariance is not lost in the process of smearing. That is, the observable is of the form $\mathsf{G}^{S(\epsilon)}$ for some generating operator $S(\epsilon)$. In fact, the operator $S(\epsilon)$ can always be expressed as a convolution of the operator $CSC^{-1}$ and the probability measure $\mu_\epsilon$, defined as the weak integral \cite{Werner}
\begin{equation}\label{convolution}
\mu_\epsilon \ast CSC^{-1} =\int W_{qp} CSC^{-1} W_{qp}^*\, d\mu_\epsilon(q,p),
\end{equation}
which is clearly a positive operator with unit trace. 
\begin{proposition}
The high-amplitude limit observable $\mu_\epsilon \ast \mathsf{G}^{CSC^{-1}}$ is a covariant phase space observable with the generating operator $\mu_\epsilon \ast CSC^{-1}$, that is, $\mu_\epsilon \ast \mathsf{G}^{CSC^{-1}}=\mathsf{G}^{\mu_\epsilon \ast CSC^{-1}}$.
\end{proposition}
\begin{proof} 
The definition of $\mu_\epsilon$ implies that $\mu_\epsilon (-Z) =\mu_\epsilon (Z)$ for all $Z\in\h B(\R^2)$, and the covariance of $\mathsf{G}^{CSC^{-1}}$ implies that $\mathsf{G}^{CSC^{-1}}_{\vp} (Z +(q,p)) =\mathsf{G}_{W_{qp}^* \vp}^{CSC^{-1}} (Z)$ for all unit vectors $\vp\in\hil$ and $ Z\in \h B(\R^2)$. Now we may use these facts, the commutativity of the convolution of probability measures, and Fubini's theorem to see that
\begin{eqnarray*}
 \langle \vp \vert (\mu_\epsilon \ast \mathsf{G}^{CSC^{-1}})  (Z) \vp\rangle &=&  (\mu_\epsilon \ast\mathsf{G}^{CSC^{-1}}_{\vp} )(Z) =\int \mathsf{G}^{CSC^{-1}}_{\vp} (Z-(q,p)) \, d\mu_\epsilon (q,p) \\
&=& \int \mathsf{G}^{CSC^{-1}}_{\vp} (Z+(q',p')) \, d\mu_\epsilon (q',p') =\int \mathsf{G}^{CSC^{-1}}_{W_{q'p'}^*\vp} (Z) \, d\mu_\epsilon (q',p')\\
&=&\frac{1}{2\pi} \int \left( \int_Z \langle \vp \vert W_{q'p'} W_{xy} CSC^{-1} W_{xy}^* W_{q'p'}^* \vp\rangle \, dxdy\right)\, d\mu_\epsilon (q',p') \\
&=& \frac{1}{2\pi} \int_Z \left( \int \langle \vp \vert W_{xy} W_{q'p'} CSC^{-1} W_{q'p'}^* W_{xy}^* \vp\rangle \, d\mu_\epsilon (q',p')\right)\, dxdy\\
&=&\frac{1}{2\pi} \int_Z \langle \vp \vert W_{xy}  (\mu_\epsilon\ast CSC^{-1} )W_{xy}^\ast \vp\rangle\, dxdy =\langle\vp \vert \mathsf{G}^{\mu_\epsilon\ast CSC^{-1} } (Z) \vp\rangle
\end{eqnarray*}
 for all unit vectors $\varphi\in\hil$ and $Z\in \h B(\R^2)$.
\end{proof}

Since there is a one-to-one correspondence between the covariant phase space observables and the generating operators, many questions concerning the  properties of a given  observable can be answered by studying only the properties of the generating operator. As an example, consider the extremality of an observable $\mathsf{G}^S$ in the sense of Holevo \cite{Holevo2}. The set of all covariant phase space observables is a convex set, and the convex combination of two observables $\mathsf{G}^{S_1}$ and $\mathsf{G}^{S_2}$ is simply $t\mathsf{G}^{S_1} +(1-t) \mathsf{G}^{S_2} =\mathsf{G}^{tS_1 +(1-t)S_2}$. Hence, an observable $\mathsf{G}^S$ is an extreme point of the convex set of covariant phase space observables if and only if $S$ is an extreme point of the set of positive trace class operators with unit trace. Furthermore, the extreme points of this set are the one-dimensional projections $P[\vp]$, $\vp\in\hil$, $\Vert\vp \Vert =1$, that is, the pure states. In the case of our specific observable, we obtain the following result.

\begin{proposition}\label{pure}
 The generating operator $\mu_\epsilon \ast CSC^{-1}$ is a pure state if and only if $S$ is a pure state and the detectors are ideal.
\end{proposition}
\begin{proof}
If $\mu_\epsilon$ is the Dirac measure concentrated at the origin, then $\mu_\epsilon \ast CSC^{-1}=CSC^{-1}$, which is a pure state if and only if $S$ is a pure state. If $\mu_\epsilon$ has the density $f_\epsilon$, then $\mu_\epsilon\ast CSC^{-1}$ can not be a pure state since this would require that $\langle \vp \vert W_{qp} CSC^{-1} W_{qp}^* \vp\rangle =1$ for all $(q,p)\in\R^2$ and some unit vector $\vp\in\hil$, which is impossible since the projective representation $(q,p)\mapsto W_{qp}$ is irreducible. Similarly in the case that  one of the measures $\mu_{\epsilon_1,\epsilon_3}$ or $\mu_{\epsilon_2 ,\epsilon_4}$ in equation \eqref{mitta} is a Dirac measure, we find that $\mu_\epsilon \ast CSC^{-1}$ is never a pure state since this would require the existence of an eigenvector of either $Q$ or $P$.
\end{proof}

The consequence of proposition \ref{pure} is that  whenever detector inefficiences are present, the measured observable can be written as a nontrivial convex combination $\mathsf{G}^{\mu_\epsilon\ast CSC^{-1} } =t\mathsf{G}^{S_1} +(1-t)\mathsf{G}^{S_2}$ for some generating operators $S_1$ and $S_2$, and for some weight factor $t\in [0,1]$. This is usually taken to correspond to classical randomization between the two observables $\mathsf{G}^{S_1}$ and $\mathsf{G}^{S_2}$. In particular, proposition \ref{pure} thus verifies the perhaps intuitive fact that ideal detectors are necessary for the measurement to be a pure quantum measurement. In the nonideal case, a natural question is whether the state $S$ itself can be a component of $\mu_\epsilon \ast CSC^{-1}$ so that the measurement of $\mathsf{G}^{\mu_\epsilon\ast CSC^{-1} } $ could be seen as a randomization of the ideal observable $\mathsf{G}^S$ with another observable $\mathsf{G}^{S'}$ which takes care of detector inefficiencies. Though the possible decompositions of a positive trace one operator into its pure components have been fully characterized \cite{Hadjisavvas, Cassinelli2}, we are unable to answer the above question in general. However, if the parameter field is in the  vacuum state, $S=\vert 0\rangle \langle 0\vert$, and  the detector efficiencies are equal, $\epsilon_j =\epsilon\in (0,1)$ for all $j=1,2,3,4$, we can easily calculate
$$
\mu_\epsilon *\vert 0\rangle \langle 0\vert =\epsilon \, \sum_{n=0}^\infty (1-\epsilon )^n \vert n\rangle\langle n \vert=\epsilon \vert 0\rangle\langle 0\vert  +(1-\epsilon )S',
$$
where $S'= \frac{\epsilon}{1-\epsilon} \sum_{n=1}^\infty (1-\epsilon )^n \vert n\rangle\langle n\vert$. In this particular case, the vacuum state is always a component of the generating operator, and thus $\mathsf{G}^{\mu_\epsilon * \vert 0\rangle\langle 0\vert} =\epsilon \mathsf{G}^{\vert 0\rangle} +(1-\epsilon) \mathsf{G}^{S'}$, a result which is already implicitly contained in \cite{Leonhardt1993}.

Usually in the process of smearing, the state distinguishing power of the observable decreases. However,  due to the Gaussian structure of the convolving measure this effect is avoided. In particular, we may prove the following result.

\begin{proposition}\label{info}
The observables $\mu_{\epsilon}\ast \mathsf{G}^{CSC^{-1}}$ and $\mathsf{G}^{CSC^{-1}}$ are informationally equivalent. In particular, $\mu_{\epsilon} *\mathsf{G}^{CSC^{-1}}$ is informationally complete if and only if $\mathsf{G}^{CSC^{-1}}$ is informationally complete.
\end{proposition}
\begin{proof}
We prove this by showing that for any state $\rho$ the densities of the corresponding probability measures can always be obtained from each other. Obviously the density of $\mu_\epsilon \ast \mathsf{G}^{CSC^{-1}}_\rho$ can always be calculated from the density of $\mathsf{G}^{CSC^{-1}}_\rho$ by performing the convolution transform, so we need to show that the convolution can always be inverted. 

First, let $\epsilon_j$, $j=1,2,3,4$, be such that $\mu_\epsilon$ has the density $f_\epsilon$. The probability measure $\mu_\epsilon \ast \mathsf{G}^{CSC^{-1}}_\rho$ has the density $f_\epsilon \ast g^{CSC^{-1}}_\rho$ which we denote by $h$. Since $\hat{h}  =2\pi \hat{f_\epsilon} \hat{g}^{CSC^{-1}}_\rho$ and the function  $\hat{f_\epsilon}$ is nonzero everywhere, it follows that if $\hat{g}^{CSC^{-1}}_\rho \in L^1(\R^2)$ then
$$
g^{CSC^{-1}}_\rho (x,y) = \frac{1}{4\pi^2} \int e^{i(xq +yp)} \frac{\hat{h} (q,p)}{\hat{f}_\epsilon (q,p)} \, dqdp
$$
for almost all $(x,y)\in \R^2$.  Using \cite[Proposition 3.4(1)]{Werner} we find that 
$$
\hat{g}^{CSC^{-1}}_\rho (q,p) = \frac{1}{4\pi^2} \int e^{-i(qx+py)} \textrm{tr}[\rho W_{xy} CSC^{-1} W_{xy}^*]\, dxdy=   \frac{1}{2\pi} \textrm{tr}[ \rho W_{p,-q} ] \textrm{tr} [W_{p,-q}^*CSC^{-1} ].
$$
Since both of the functions on the right-hand side are square-integrable, it follows that their product is integrable, and thus $\hat{g}^{CSC^{-1}}_\rho\in L^1 (\R^2)$.

If either $\epsilon_1 =\epsilon_3 =1$ or $\epsilon_2=\epsilon_4 =1$, then the density is of the form $h_j =f_j \ast g^{CSC^{-1}}$, $j=1,2$, where $f_j$ is a one-dimensional Gaussian and the convolution is taken only with respect to the first or second argument depending on the case in question.  The  Fourier transform now gives $\hat{h}  =\sqrt{2\pi} \hat{f_j} \hat{g}^{CSC^{-1}}$, where $\hat{f_j}$ is the one-dimensional Fourier transform of $f_j$. Thus, the same argument holds in these special cases, and we may conclude that the convolution can always be inverted.

\end{proof}

Apart from showing the informational equivalence of the observables, the proof of proposition \ref{info} also provides a practical means of compensating the additional smoothing caused by detector inefficiencies. This means that the measurement statistics of the ideal observable can always be extracted from the statistics obtained by inefficient measurements. This is of particular importance when applying this measurement scheme in quantum state reconstruction. Suppose that the parameter field is in the vacuum state. With ideal detectors, this measurement setup constitutes a measurement of the observable $\mathsf{G}^{\vert 0\rangle}$, and the density of the corresponding probability measure $\mathsf{G}^{\vert 0\rangle}_\rho$ is the $Q$-function of the signal state $\rho$. It is a well-known fact that this observable is informationally complete, or in other words, the $Q$-function determines the state uniquely. Furthermore, several reconstruction formulae for calculating the matrix elements of $\rho$ with respect to the number basis are also known. (see, e.g., \cite{KiPeSc, Paris, Richter}). Thus, in the case of inefficient detectors we may always apply the method of proposition \ref{info} to first reconstruct the $Q$-function from the measurement statistics and then proceed to determine the state of the signal field.

\

\section{Conclusion}\label{conclusions}
We have considered the eight-port homodyne detection scheme in the case that each of the associated photon detectors is assigned with a different quantum efficiency. We have shown that in the high-amplitude limit, the measured observable approaches a covariant phase space observable which is a smearing of the one obtained by using ideal detectors. We have also studied some properties of the high-amplitude limit observable. In particular, we have shown that the state distinguishing power of the observable does not depend on the quantum efficiencies of the detectors. Futhermore, we have seen that when detector inefficiencies are present, the measured observable is never extremal. That is, the measurement is never a pure quantum measurement.

To conclude, we wish to emphasize that the quantum efficiency is only one of the properties which characterize a photon detector. In fact, there is a wide variery of features, ranging from operating temperatures to dark count rates, which are used to classify the detectors  \cite{Hadfield}. Thus, in the case of a specific measurement one needs to decide which properties are the ones that need to be optimized. The consequence of proposition \ref{info} is that in eight-port homodyne detection, the quantum efficiencies have no effect on the amount of information obtained about the state of the signal field. Hence, at least for the purpose of quantum state reconstruction, the quantum efficiencies of the available photon detectors are of little relevance, and one may concentrate on the other properties of the detectors.

\

\noindent {\bf Acknowledgment.} J. S. was supported by the Turku University Foundation and the Finnish Cultural Foundation during the preparation of the manuscript.

\end{document}